\documentclass[10pt,twocolumn,twoside]{IEEEtran}
\usepackage{amsmath,amsfonts}
\usepackage{algorithm}
\usepackage{algpseudocode}
\usepackage{amsthm}
\usepackage{array}
\usepackage[caption=false,font=normalsize,labelfont=sf,textfont=sf]{subfig}
\usepackage{textcomp}
\usepackage{stfloats}
\usepackage{url}
\usepackage{verbatim}
\usepackage{graphicx}
\usepackage{cite}
\usepackage{tikz}
\usepackage{dsfont}
\usepackage[mathscr]{euscript}
\usepackage{booktabs}
\usepackage{makecell}
\usepackage{cuted}

\newcommand{\BibTeX}{\rm B\kern-.05em{\sc i\kern-.025em b}\kern-.08em\TeX}

\usetikzlibrary{shapes.geometric, arrows, calc, positioning}

\newtheorem{assumptions}{Assumptions}
\newtheorem{definition}{Definition}
\newtheorem{theorem}{Theorem}
\newtheorem{proposition}{Proposition}
\newtheorem{lemma}{Lemma}

\DeclareMathOperator*{\argmax}{argmax}
\DeclareMathOperator*{\argmin}{argmin}

\definecolor{rougeINRIA}{RGB}{230, 51, 18}
\definecolor{grisbleuINRIA}{RGB}{56,66,87}
\definecolor{orangeINRIA}{RGB}{240, 126, 38}
\definecolor{lilasINRIA}{RGB}{155,0,79}
\definecolor{bleuINRIA}{RGB}{20,136,202}
\definecolor{vertINRIA}{RGB}{149,193,31}
\definecolor{jauneINRIA}{RGB}{255,205,28}
\definecolor{mauveINRIA}{RGB}{101,97,169}
\definecolor{bleuclairINRIA}{RGB}{137,204,202}
\definecolor{vertclairINRIA}{RGB}{199,214,79}

\newcommand*{\defeq}{:=}

\begin{document}

\title{Game Theory and Multi-Agent Reinforcement Learning for Zonal Ancillary Markets}

\author{Francesco Morri\thanks{Corresponding author: francesco.morri@inria.fr}, 
        Hélène Le Cadre,
        Pierre Gruet\thanks{P. Gruet is with EDF R$\&$D and FiME laboratory, Paris, France.},
        Luce Brotcorne\thanks{F. Morri, H. Le Cadre and L. Brotcorne are with Inria, Univ. Lille, CNRS, Centrale Lille, UMR 9189 CRIStAL, F-59000 Lille, France.}}

\markboth{PREPRINT}%
{F. Morri \MakeLowercase{\textit{et al.}}: Game Theory and Multi-Agent Reinforcement Learning for Zonal Ancillary Markets}


\maketitle

\begin{abstract}
We characterize zonal ancillary market coupling relying on noncooperative game theory. We formulate the ancillary market as a multi-leader single follower Stackelberg game, that we subsequently cast as a generalized Nash game with side constraints and nonconvex feasible sets. We determine conditions for the existence of an equilibrium and show that the game has a generalized potential structure. To compute market equilibrium, we rely on two exact approaches: an integrated optimization reformulation involving complementarity constraints, that can be solved using a nonlinear optimization solver, and Gauss-Seidel best-response. We compare these exact methods against multi-agent deep reinforcement learning (MARL). Simulations on real data from Germany and Austria ancillary markets show that MARL better capture intricate interactions between market agents under varying information-sharing scenarios, achieving faster learning rates to reach a stationary solution, than both exact approaches. Further, MARL results in smaller market costs compared to the exact methods, at the cost of higher variability in the profit allocation among stakeholders. Finally, on the policy side, we show numerically that tighter coupling between adjacent supply zones reduces the average market cost of the largest zone.
\end{abstract}


\section{Introduction}
A characteristic feature of electricity is the impossibility to store it in large proportions. Hence, it must be produced exactly when it is consumed, which raises a need for a real-time monitoring of the balance between production and consumption to guarantee the stability of the power grid. When imbalances occur, mechanisms are automatically triggered to restore the nominal frequency of the electric system: a change in power production is made by some production units, which are then said to provide \emph{ancillary services}, in exchange of a monetary contribution.
In Western Europe, many market zones are progressively switching from an automatic and mandatory involvement of qualified production units with an amount that is fixed by the system operator to a market mechanism based on daily auctions in which the market operator publishes a demand and then the producers submit offer curves that are aggregated until they meet the demand. From the point of view of a market actor, submitting a bid in such an auction leaves some space for strategic bidding, which requires information about the other market participants. This information, together with the design of a relevant bidding strategy, should enable market actors to engage their production assets in the best way in ancillary services markets, which will then give them a fair idea of the revenue they can expect.
Although these ancillary services markets exist in many countries, the bidding zones can be somehow associated: for instance the Austrian and German areas can share their bids, subject to some constraints and to the availability of the transmission network across both countries. Market participants might be willing to investigate the impact of this shared market structure and to see if it stabilizes their revenues over time.
The main challenges to address are the characterization of the market equilibria and the understanding of their evolution when market zones get interconnected: from the perspective of the market operator, it is important to gain insight into the behaviors of the stakeholders and to gauge the opportunity to enforce market integration over multiple countries if this can, e.g., mitigate price peak events in one country thanks to interconnection.
In this article, we use noncooperative game theory to model ancillary services markets. Market equilibria are computed relying on reformulations of the game and multi-agent deep reinforcement learning (MARL). Our flexible framework allows us to assess the consequences of the interconnection of market zones by dynamically changing the coupling factors. The approach can be generalized to any number of zones.
\IEEEpubidadjcol
\subsection{Related Work}
\label{sec:related}
There are two main directions for research on ancillary (and zonal) market design: a) working on the design of the market itself, studying different formulations and their features \cite{lecadre2019}; b) working on the strategic bidding aspect \cite{morri2024}, focusing more on the point of view of the energy producers.
In \cite{oren_design_2001}, Oren discusses possible ways of clearing the electricity market and selecting bids, proposing different options for the market operator's objective function, the settlement rules and the pricing of the products. Previous works are mostly theoretical, comparing the efficiency of different market designs relying on noncooperative game theory. A new market design is proposed by Ela et al. in \cite{ela_effective_2012}, taking into account the uncertainty resulting from the massive penetration of wind power in the grid. 
In the same vein, \cite{ela_market_2014,ela_market_2014_2} discuss the design of an ancillary services market focused on primary frequency response, and then simulate this design under different scenarios. Another example exploring the impact of renewable based generators is \cite{merten2020bidding}, which tackles the optimization of bidding strategies for ancillary markets involving battery energy storage systems. We conclude the discussion on the formulation and features of ancillary services markets with \cite{oureilidis_ancillary_2020}, which explores the pros. and cons. of ancillary services market implementation at the distribution grid level, considering frequency response, voltage power regulation and other grid related parameters.
Regarding the computation of bidding strategies and solutions for ancillary and zonal markets, a first contribution can be found in \cite{morinec2008optimal}, where the authors model the market as a two-player noncooperative game, analyzing the strategies of the two producers, and solving the game to obtain different market equilibria. In \cite{meyn_ancillary_2014}, a control based decentralized solution is proposed to create ancillary services helping with power imbalances caused by renewable generation sources. Sarfati and Holmberg in \cite{sarfati_simulation_2020} formulate a zonal market as a two-stage stochastic problem, which they then recast as a mixed-integer bilinear program that they subsequently solve relying on a nonlinear optimization solver. This solution is computed from an external point of view, considering all producers at once. In \cite{jay2020game}, the market is modeled as a Stackelberg game where the market operator is the leader and the producers are the followers; the solution is then obtained by recasting the problem as a mathematical program with equilibrium constraints. Finally, \cite{deng_distributed_2024} focuses on a game-theoretic model with nonlinear convex feasible sets, which in turn is solved by simulation relying on a linear dynamic model. The authors then analyze theoretically the existence of a variational equilibrium and propose a distributed algorithm to compute it.
The use of learning algorithm for strategic bidding has become more common in the last years. In \cite{ye_deep_2020}, a day-ahead electricity market bidding problem with multiple strategic producers is formulated as a Markov game with unobservable information, which is solved relying on a model-free and data-driven approach. Regarding ancillary markets, the application of learning based algorithm is not particularly widespread, nonetheless, there are some examples. In \cite{abgottspon_strategic_2013}, Abgottspon and Andersson focus on ancillary markets with hydro-based producers. Their simulation approach uses discrete inputs and discretized bids for the learning producers. The results are not definitive, but the learning agents show sign of strategic behavior. A more advanced approach can be found in \cite{du_approximating_2021}, where the implemented learning algorithm utilizes continuous inputs and outputs and is specifically designed to handle multi-agent settings. The producers here only bid a price, instead of a price-quantity couple as in our case, and the focus is on the grid optimization.
\subsection{Contributions and Organization}
Our main contributions are summarized as follows:
\begin{itemize}
    \item we formulate a multi-zone ancillary market model as a Stackelberg game, that we subsequently reformulate as a generalized Nash game with generalized potential structure. We obtain theoretical guarantees for the existence of market equilibria;
    \item we propose three algorithms to compute market outcomes: an integrated optimization approach and Gauss-Seidel best-response as exact methods; a multi-agent deep reinforcement learning (MARL) approach, that can be implemented in a decentralized fashion, reaching a stationary solution;
    \item we compare the three algorithm performance on real data from the Germany-Austria market, comparing social cost, producers' profit variability and computation times. Finally, comparing multiple coupling scenarios, we observe a decrease in the largest zone market cost with tighter couplings.
\end{itemize}
This paper is organized as follows: in Section~\ref{sec:problem}, we formulate the ancillary services market as a noncooperative game, in Section~\ref{sec:analysis}, we derive conditions for existence of market equilibria, that can be subsequently computed relying on the algorithms introduced in Section~\ref{sec:algorithms}. Finally, in Section~\ref{sec:simul}, we implement the algorithms on real data from Austria and Germany, that we compare numerically using key performance metrics. We conclude in Section~\ref{sec:conclusions}.

\subsection*{Notations}
Given $\mathcal{A} \subset \mathbb{R}^n$, $|\mathcal{A}|$ refers to the cardinality of the set.
Let $x_n \in \mathbb{R}^{m_n}$, with $m_n \in \mathbb{N}^\star$, we define $\text{col}((x_n)_n)$ as the stack of the $(x_n)_n$ vectors.
Given a vector of decision variables $x:=(x_n)_{n\in\mathcal{N}}$, we define $x_{-n} \defeq (x_m)_{m\in\mathcal{N}\setminus\{n\}}$ as the concatenation of all the players' decision variables except $n$. 
When we write $\nabla^2 f(x^\star)$, with $x^\star\in\mathbb{R}^d$, $d \in \mathbb{R}^\star$, we consider the Hessian matrix of the function $f$ evaluated in $x^\star$, where the element $(i,j)$ refers to $\partial_{x_i,x_j}f(x^\star),\,\forall i,j\in \{1,\hdots,d\}$.
\begin{table}[t]
\caption{List of variables, index and sets.}
    \label{tab:notation}
    \centering
    \begin{tabular}{c|c}
    \toprule
         Symbol & Meaning \\ \midrule
         $\mathcal{N}$ & Set of $N$ producers \\
         $\mathcal{Z}$ & Set of $Z$ zones \\
         $z(n)$ & Zone of producer $n$\\
         $\mathcal{N}(z)$ & Subset of producers in $z$ \\
         $K_n$ & Maximum number of bids for producer $n$\\
         $B$ & Total number of bids \\
         $\Delta_{n,k}$ & Capacity offered by producer $n$ for bid $k$\\
         $\underline{\Delta}_n,\bar{\Delta}_n$ & Lower and upper capacity for producer $n$\\
         $\pi_{n,k}$ & Price offered by producer $n$ for bid $k$\\
         $\underline{\pi}_n,\bar{\pi}_n$ & Lower and upper price for producer $n$\\
         $x_{n,k}^z$ & Fraction of accepted bid $k$ of producer $n$ in zone $z$\\
         \bottomrule
    \end{tabular}
\end{table}
\section{An Adversarial Game Formulation}
\label{sec:problem}
The ancillary services market can be formulated as a multi-leader single follower Stackelberg game. At the upper level, producers (leaders) bid strategically to maximize their revenues while anticipating the market clearing. At the lower level, the market operator (follower) determines the fractions of the bids that are accepted in order to minimize the total cost it has to pay to the producers to meet the zonal demands, while satisfying zonal, interzone and bid-related constraints. Because of the opposite objectives of the producers and market operator, the problem can be interpreted as an adversarial game, also called a zero-sum game.

We define $\mathcal{N}\defeq\{1, \hdots, N\}$ as the set of $N$ producers and $\mathcal{Z}\defeq\{1,...,Z\}$ as the set of zones.
Each producer $n \in \mathcal{N}$ can bid up to $K_n \in \mathbb{N}^\star$ bids. Thus, the maximum number of bids is $B\defeq\sum_{n\in\mathcal{N}}K_n$. Further, each zone $z \in \mathcal{Z}$ has three input parameters: the total requested demand $D_z \geq 0$, the export limit $E_z \geq 0$ and the core portion $C_z \geq 0$. For all $n\in\mathcal{N}$, $z(n)\in\mathcal{Z}$ defines the zone of producer $n$ and $\mathcal{N}(z)\subseteq\mathcal{N}$ is the subset of producers in zone $z$. Each producer $n$'s decision variable is a finite sequence of bids 
\begin{equation}\label{eq:full_prod_vars}
y_n \defeq \text{col}(\Delta_n,\pi_n) = \text{col}\Big((\Delta_{n,l},\pi_{n,l})_{l=1,\hdots,K_n}\Big),
\end{equation}
where $\Delta_n,\pi_n$ are producer $n$'s finite sequences of offered capacities and prices, respectively. We denote by $y\defeq (y_n)_{n}$ the producers' strategy profile. 

The decision variables of the market operator are $x\defeq(x_{n,k}^z)_{n,k,z} \in [0,1]^{B\times Z}$, where $x_{n,k}^z$ represents the fraction of bid activated in each zone. A summary of the definitions is displayed in Tab.~\ref{tab:notation}.

\subsection{Producers}
Each producer $n \in \mathcal{N}$ has a capacity $\bar{\Delta}_n\geq 0$ and a marginal price $\underline{\pi}_n \geq 0$, and can either submit a single price (with full capacity) or a finite sequence of price and quantity pairs satisfying the following constraints involving finite upper and lower bounds:
\begin{subequations}
    \begin{align}
        &  \underline{\pi}_n\leq \pi_{n,k}\leq \bar{\pi}_n,\quad \forall\; k\in\{1,\dots,K_n\},\label{eq:price_bound}\\
        & \underline{\Delta}_n\leq\Delta_{n,k}\leq\bar{\Delta}_n,\quad\forall\; k\in\{1,\dots,K_n\},\label{eq:capacity_bound}\\
        &  \sum_{k=1}^{K_n}\Delta_{n,k} \leq \bar{\Delta}_n.\label{eq:tot_capacity_bound}
    \end{align}
    \label{eq:producer_const}
\end{subequations}
We define producer $n$'s feasible set as $\mathcal{Y}_n \defeq \{y_n\in \mathbb{R}^{2\times K_n}\,|\, y_n \mbox{ satisfies Eqs.\,\eqref{eq:producer_const}} \}$. The objective of producer $n$ is to maximize its revenue defined as follows:
\begin{equation}
    J_n(x_n,y_n) \defeq \sum_{k=1}^{K_n}\pi_{n,k}\Delta_{n,k}\left(\sum_{z\in \mathcal{Z}} x_{n,k}^z\right).
    \label{eq:J_n}
\end{equation}

\subsection{Market Operator}
The market operator deals with four sets of constraints that depend on the capacities $\Delta_{n,k}, \forall k\in[1,K_n],n\in\mathcal{N}$ submitted by the producers\\
\textbf{Zonal Demand}
\begin{equation}
    \sum_{n\in\mathcal{N}}\sum_{1\leq k\leq K_n}x_{n,k}^z\Delta_{n,k}\geq D_z,\quad \forall z\in\mathcal{Z}. \label{eq:demand_c}
\end{equation}
\textbf{Capacity Limit}
\begin{equation}
    \sum_{z\in\mathcal{Z}} x_{n,k}^z\leq 1,\quad \forall k\in[1,\dots,K_n],\quad\forall n\in\mathcal{N}. \label{eq:accepted_c}
\end{equation}
\textbf{Export Limit}
\begin{equation}
    \sum_{n\in\mathcal{N}(z)}\sum_{1\leq k\leq K_n}\sum_{z'\neq z}x_{n,k}^{z'}\Delta_{n,k} \leq E_z,\quad \forall z\in\mathcal{Z}. 
    \label{eq:export_c}
\end{equation}
\textbf{Core Portion}
\begin{equation}
    \sum_{n\in\mathcal{N}(z)}\sum_{1\leq k\leq K_n} x_{n,k}^z\Delta_{n,k}\geq C_z,\quad \forall z\in\mathcal{Z}.
    \label{eq:core_c}
\end{equation}
Eq.~\eqref{eq:demand_c} ensures that the total demand in each zone is met, Eq.~\eqref{eq:core_c} ensures that there is a given amount of the demand of a specific zone satisfied by the producers in that zone, while Eq.~\eqref{eq:export_c} gives a cap to the amount of energy that can be exported.
Finally, Eq.~\eqref{eq:accepted_c} captures the fact that for a given bid $(n,k)$, the sum of the fractions activated in each zone is normalized. To highlight the mathematical structure of the problem, constraints \eqref{eq:demand_c}-\eqref{eq:core_c} can be compactly written under the form
$$\Phi(x,y) \leq 0,$$ where $y:=(y_n)_n$, $y_n$ is defined in Eq.~\eqref{eq:full_prod_vars}, and $\Phi:[0,1]^{B\times Z}\times\prod_{n\in\mathcal{N}}[[\underline{\Delta}_n,\bar{\Delta}_n]^{K_n}\times[\underline{\pi}_n,\bar{\pi}_n]^{K_n}]\to \mathbb{R}^{3Z+N}$. We denote $$\mathcal{X}(y):=\{x\in[0,1]^{B\times Z} \,|\, \Phi(x,y)\leq0\},$$ the feasible set of the market operator which depends on the profile of strategies of all the producers $y:=(y_n)_n$.
The market operator aims to minimize the total cost it has to pay to the producers. Its objective function takes the form
\begin{equation}
    J_{\text{MO}}(x,y) \defeq \sum_{n\in\mathcal{N}}\sum_{1\leq k\leq K_n}\left(\sum_{z\in\mathcal{Z}}x_{n,k}^z\right)\Delta_{n,k}\pi_{n,k}.
    \label{eq:J_mo}
\end{equation}
The market operator determines the fractions of bids to activate in each zone $(x_{n,k}^z)_{n,k,z}$, under constraints \eqref{eq:demand_c}-\eqref{eq:core_c}
\begin{subequations}
    \begin{align}
    \min_{x} & \quad J_{\text{MO}}(x,y), \label{eq:market_obj}\\
    \text{s.t.} & \quad x\in\mathcal{X}(y).
\end{align}
\label{eq:market_clearing}
\end{subequations}
We report a schematic view of a two-zone market in Fig.~\ref{fig:market_diagram}, where we highlight with arrows the different constraints.
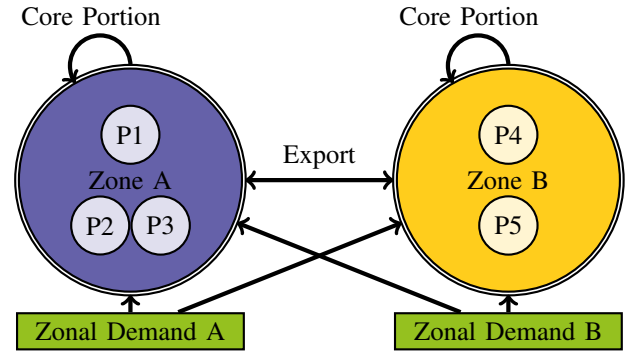
\begin{figure}[t]
    \centering
    \begin{tikzpicture}[scale=1.]
        \coordinate (na) at (0,0);
        \coordinate (nb) at (5,0);
        \node[circle, thick, double, fill=mauveINRIA, draw, minimum size=3cm] (a) at (na) {Zone A};
        \node[circle, thick, double, fill=jauneINRIA, draw, minimum size=3cm] (b) at (nb) {Zone B};
        \node[rectangle, thick, fill=vertINRIA, draw, minimum width=3cm] (d_a) at ($ (na) + (0,-2) $) {Zonal Demand A};
        \node[rectangle, thick, fill=vertINRIA, draw, minimum width=3cm] (d_b) at ($ (nb) + (0,-2) $) {Zonal Demand B};
        \draw[<->, ultra thick] (a) -- (b) node[midway, above] {Export};
        \draw[ultra thick,->] (a.90) arc (0:215:4mm) node[midway, above] {Core Portion};
        \draw[ultra thick,->] (b.90) arc (0:215:4mm) node[midway, above] {Core Portion};
        \draw[ultra thick,->] (d_a) -- (a);
        \draw[ultra thick,->] (d_b) -- (b);
        \draw[ultra thick,->] (d_a) -- (b);
        \draw[ultra thick,->] (d_b) -- (a);

        \node[circle, thick, fill=mauveINRIA!20, draw, minimum size=0.5cm] (a1) at ($ (na) + (0,0.6) $) {P1};
        \node[circle, thick, fill=mauveINRIA!20, draw, minimum size=0.5cm] (a2) at ($ (na) + (-0.4,-0.6) $) {P2};
        \node[circle, thick, fill=mauveINRIA!20, draw, minimum size=0.5cm] (a3) at ($ (na) + (0.4,-0.6) $) {P3};

        \node[circle, thick, fill=jauneINRIA!20, draw, minimum size=0.5cm] (b1) at ($ (nb) + (0,0.6) $) {P4};
        \node[circle, thick, fill=jauneINRIA!20, draw, minimum size=0.5cm] (b2) at ($ (nb) + (0,-0.6) $) {P5};
    \end{tikzpicture}
    \caption{Representation of an ancillary market with two zones, with three producers in the first zone and two producers in the second zone. Constraints are highlighted with arrows.}
    \label{fig:market_diagram}
\end{figure}

\subsection{Market Clearing and KKT Reformulation}
We formulate the ancillary services market clearing problem  as a pessimistic multi-leader single follower Stackelberg game
\begin{subequations}
\begin{alignat}{3}
    &\forall n\in \mathcal{N}, \quad && \min_{x_n} \max_{y_n}&& \quad  J_n(x_n,y_n), \label{eq:MO:1} \\
    & && \text{s.t.} && x \in \argmin_{x \in \mathcal{X}(y)} \quad J_{\text{MO}}(x,y), \label{eq:MO:2}\\
    & && && y\in \prod_j\mathcal{Y}_j. \label{eq:MO:3}
\end{alignat}
\label{eq:full_problem}
\end{subequations}
With this formulation, each producer optimizes its bids without considering the different zones. Such a distinction is made by the market operator when it determines the fraction of bids accepted in each zone. Further, each producer's objective function is independent of the other producers' strategies. Coupling between the objective functions occurs through the market operator's strategy. In case of multiple solutions at the lower level, the pessimistic formulation enables the selection of the worst equilibrium for the producers, obtaining a more robust formulation of the market.

The Lagrangian function associated with the market clearing defined in problem \eqref{eq:market_clearing} takes the form
\begin{equation}
    L(x,y;\Lambda) = \sum_{n\in\mathcal{N}}\sum_{k=1}^{ K_n}\left(\sum_{z\in\mathcal{Z}}x_{n,k}^z\right)\Delta_{n,k}\pi_{n,k} + \Lambda^\top\Phi(x,y),
    \label{eq:lagrangian}
\end{equation}
where $\Lambda \in \mathbb{R}_{+}^{3Z+B}$ refers to the dual variables associated with constraints~\eqref{eq:demand_c}-\eqref{eq:core_c}. 
\begin{lemma}[Slater's Constraint Qualification] \label{lem:slater}
    Slater's constraint qualification holds if one of these conditions is satisfied
    \begin{itemize}
        \item[(i)] $\displaystyle{\sum_{n\in \mathcal{N}(z)}}\bar{\Delta}_n > D_z,\;\forall z\in\mathcal{Z}$,
        \item[(ii)] $\displaystyle{\sum_{z'\neq z}}\left(\sum_{n\in \mathcal{N}(z')}\bar{\Delta}_n-D_{z'}\right) > D_z - \sum_{n\in \mathcal{N}(z)}\bar{\Delta}_n$ \\
        and \\$D_z - \sum_{n\in \mathcal{N}(z)}\bar{\Delta}_n\leq \sum_{z'\neq z}E_{z'}$.
    \end{itemize}
\end{lemma}
\begin{proof}
Under condition (i), each zone can satisfy its own demand; thus, there exists an interior-point solution. Under condition (ii), one zone $z \in \mathcal{Z}$ may not be able to meet its own demand, i.e. $D_z - \sum_{n\in \mathcal{N}(z)}\bar{\Delta}_n>0.$

In this situation, we need enough demand offers left from other zones to cover the lack of supply, thus $$\sum_{z'\neq z}\left(\sum_{n\in \mathcal{N}(z')}\bar{\Delta}_n-D_{z'}\right) > D_z - \sum_{n\in \mathcal{N}(z)}\bar{\Delta}_n.$$ Since the export constraint should hold, we get that $D_z - \sum_{n\in \mathcal{N}(z)}\bar{\Delta}_n\leq \sum_{z'\neq z}E_{z'}.$
\end{proof}
When the lower-level problem \eqref{eq:market_clearing} satisfies Lemma~\ref{lem:slater}, we can replace it by its KKTs in \eqref{eq:full_problem}. Hence, each producer has to anticipate the decision $x$ of the market operator, meaning that in the single level the decision variables for producer $n$ will be $y_n \in \mathcal{Y}_n$ and $x_n \in \mathcal{X}_n(y)$ with $y=(y_n)_n$. 
Before writing the full reformulation, we introduce tracking variables for every zone $z\in \mathcal{Z}$
\begin{subequations}
    \begin{align}
        & d^z_n(x_{-n},y_{-n}) \defeq D_z - \sum_{j\neq n}\sum_{1\leq k\leq K_j}x_{j,k}^z\Delta_{j,k},\\
        & e^z_n(x_{-n},y_{-n}) \defeq E_z - \sum_{\substack{j\in \mathcal{N}(z),\\j\neq n}} \sum_{z'\neq z}\sum_{1\leq k\leq K_j} x_{j,k}^{z'}\Delta_{j,k},\\
        & c^z_n(x_{-n},y_{-n}) \defeq C_z - \sum_{\substack{j\in \mathcal{N}(z),\\j\neq n}} \sum_{1\leq k\leq K_j} x_{j,k}^z\Delta_{j,k}.
    \end{align}
\end{subequations}
The single-level problem for agent $n$ takes the form
\begin{subequations}
\begin{align}
    &\min_{x_n}\,\max_{y_n,\Lambda_n} J_n(y_n,x_n), \label{eq:single_level_obj}\\
    &\mbox{s.t. } y_n\in \mathcal{Y}_n, x_n\in \mathcal{X}_n(y), \Lambda_n\geq0,\\
    & \nabla_{x_n} L(x,y;\Lambda) = 0,\label{eq:KKT_market}\\
    &\lambda_n^z\Bigg(\sum_{1\leq k\leq K_n}x_{n,k}^z\Delta_{n,k} - d_n^z(x_{-n},y_{-n})\Bigg) = 0,\forall z\in\mathcal{Z},\label{eq:slackness_demand}\\
    & \mu_{n,k}\Bigg(\sum_{z\in\mathcal{Z}}x_{n,k}^z - 1\Bigg) = 0,\; \forall k\in\{1,\dots,K_n\},\label{eq:slackness_accepted}\\
    & \sigma_n\Bigg(\sum_{k=1}^{K_n}\Big(\sum_{z'\neq z(n)} x_{n,k}^{z'}\Big)\Delta_{n,k}-e_n^{z(n)}(x_{-n},y_{-n})\Bigg) = 0,\label{eq:slackness_export}\\
    & \delta_n\Bigg(\sum_{k=1}^{K_n}x_{n,k}^{z(n)}\Delta_{n,k} - c_n^{z(n)}(x_{-n},y_{-n})\Bigg) = 0.\label{eq:slackness_core}
\end{align}
\label{eq:single_level_problem}
\end{subequations}
Eq.~\eqref{eq:KKT_market} is the stationarity condition of the lower-level problem \eqref{eq:MO:2}-\eqref{eq:MO:3}. Though there is no explicit coupling between the producers' objective functions, Eqs.~\eqref{eq:demand_c}, \eqref{eq:export_c}, \eqref{eq:core_c} and the slackness conditions in Eqs.~\eqref{eq:slackness_demand}, \eqref{eq:slackness_export}, \eqref{eq:slackness_core} impose couplings through bids and zones. 

We now aim to highlight the structure of the problem. We first note that, given the dependence of the feasible set $\mathcal{X}$ on $y$, the decision variable $x_n$ can be expressed as an implicit function in $y$, $x_n(y)$. This allows to write the producer's objective function as an explicit function of all producers' decision variables, i.e., $J_n(y_n,y_{-n})$.

However, in this reformulation, the dual variables are decision variables of the producers. To simplify, we need an additional change of variable. We set $$w_n\defeq(\pi_{n,k}, \Delta_{n,k},x_{n,k}^z,\lambda_n^z,\mu_{n,k},\sigma_n,\delta_n) \in \mathcal{W}_n,$$ where $\mathcal{W}_n \defeq [\underline{\pi}_n,\bar{\pi}_n]^{K_n}\times[\underline{\Delta}_n,\bar{\Delta}_n]^{K_n}\times[0,1]^{K_n\cdot Z}\times[0,+\infty]^Z\times[0,+\infty]^{K_n}\times[0,+\infty]\times[0,+\infty]$ is a polyhedron of dimension $3K_n+Z(K_n+1)+2$.

The set of the $l_n$ local constraints  are represented by the function $h_n:\mathcal{W}_n\rightarrow \mathbb{R}^{l_n}$. Further, we introduce the local feasible set for each player $n$ 
$$\widetilde{\mathcal{W}}_n\defeq\{w_n\in \mathcal{W}_n\, | \, h_n(w_n)\leq 0\}.$$
Similarly, the $r$ shared constraints are formulated in a compact form as $G:\Omega(w) \rightarrow \mathbb{R}^r$, where $\Omega_n(w_{-n}):=\{w_n\in\widetilde{\mathcal{W}}_n\,|\,G(w_n,w_{-n})\leq 0\}$ and $\Omega(w) \defeq \prod_{n=1}^N\Omega_n(w_{-n})$. We check that $h_n$ and $G$ are continuously differentiable mappings. 

Relying on these notations, the $N$-producer coupled optimization problems, derived from Eqs.~\eqref{eq:single_level_problem}, give rise to a nonconvex game with side constraints \cite{pang_nonconvex_2011} 
$$\Gamma \defeq \Big(\mathcal{N},\Omega(\cdot),\left(J_n(\cdot)\right)_n\Big),$$ where producers solve the parametrized optimization problems
\begin{subequations}
    \begin{alignat}{2}
       \forall n \in \mathcal{N}, \quad & \max_{w_n\in \mathcal{W}_n} && J_n(w_n,w_{-n}),\\
        & \mbox{s.t. } && h_n(w_n) \leq 0,\\
        & && G(w_n,w_{-n})\leq 0.
    \end{alignat}
    \label{eq:GNG}
\end{subequations}
Pang and Scutari in \cite{pang_nonconvex_2011} study the properties of problem~\eqref{eq:GNG}, providing different sets of assumptions to be used in order to obtain existence guarantees for specific equilibria. In the next section, we rely on a result guaranteeing the existence of a Nash equilibrium for a relaxed version of the problem. Further, under mild assumptions on the model parameters, we show the equivalence between the set of solutions of the relaxed version and the generalized Nash game $\Gamma$. Refinements of generalized Nash equilibria (GNEs) solutions of $\Gamma$ requiring the uniform evaluation of the Lagrange multipliers associated with the coupling constraints across the producers, are called variational equilibria (v-GNEs) \cite{harker1991}. Interestingly, the generalized potential structure of $\Gamma$, guarantees the economic efficiency of v-GNEs. However, reaching v-GNEs requires some implicit coordination between the market participants which may be difficult to enforce. This will be discussed further in Section~\ref{sec:simul}.

\section{Theoretical Framework}
\label{sec:analysis}
In a large part of the literature, the convexity of the agents' optimization problems is an essential assumption under which noncooperative games are studied and analyzed. On the contrary, in this paper, $\Gamma$ is a generalized Nash game with side constraints, involving trilinear objective functions and nonconvex constraints; thus, giving rise to a nonconvex game. In order to proceed, we will rely on results established in \cite{pang_nonconvex_2011}, where the authors study a class of problems with similar settings. To that purpose, we first apply the generic setting of \cite{pang_nonconvex_2011}. Then, we present a relaxed version of our problem, which is more tractable for this analysis, providing results regarding the equivalence with the non-penalized version. Further, we prove existence of market equilibrium under mild assumptions, and highlight the generalized potential structure of $\Gamma$.

\subsection{Nonconvex Game with Side Constraints}
The formulation of \cite{pang_nonconvex_2011} considers a problem where the coupling constraints are treated as penalties with associated prices, in a modified objective function. This problem is then solved as an $(N+1)$-player game, where $N$ players optimize their variables and an external one optimizes the prices, which are shared by all players. In our formulation, the coupling constraints are Eqs.~\eqref{eq:demand_c}-\eqref{eq:export_c}-\eqref{eq:core_c}-\eqref{eq:slackness_demand}-\eqref{eq:slackness_export}-\eqref{eq:slackness_core}, while the remaining part is considered as \textit{local constraints}.

With the compact notation previously introduced, we can write the problem as an $(N+1)$-player game with prices
\begin{equation}
    \forall n \in \mathcal{N}, \quad \max_{w_n\in \widetilde{\mathcal{W}}_n}\left[ J_n(w_n,w_{-n}) - \sum_{j=1}^r\rho_jG_j(w) \right],
    \label{eq:priced_game}
\end{equation}
where we call the externally optimized price, $\rho\in\mathbb{R}^r_+$. We can now recall the definition of Nash Equilibrium.
\begin{definition}[Nash Equilibrium with Side Constraints, {\cite[Definition 1]{pang_nonconvex_2011}}]
\label{def:NE_priced}
    A Nash equilibrium (NE) for the game \eqref{eq:priced_game} is a strategy-price tuple $(w^\star,\rho^\star)$, such that for all $n$
    \begin{equation*}
        w_n^\star\in\argmax_{w_n\in \widetilde{\mathcal{W}}_n}\left[J_n(w_n,w_{-n}^\star) - \sum_{j=1}^r\rho_j^\star G_j(w_n,w_{-n}^\star)\right],
    \end{equation*}
    and
    \begin{equation*}
        0\leq\rho^\star \perp G(w^\star)\leq 0.
    \end{equation*}
\end{definition}

\subsection{Relaxed Generalized Nash Game}
\label{sec:penal_gng}
In order to extend the results from \cite{pang_nonconvex_2011}, we consider a relaxed version of our problem. We will use this relaxed version to prove the existence of an equilibrium under mild assumptions. Concretely, we add the constraints from Eqs.~\eqref{eq:demand_c} and \eqref{eq:core_c} to the objective function, which gives the function $J_n^p(w)$, the new constraint maps $h_n^p(w_n)$ and $G^p(w)$, and the associated set of penalties $M_n$ for each agent.
We can formulate the penalized generalized Nash game $\Gamma^p$ as follows:
\begin{subequations}
    \begin{alignat}{2}
        & \max_{w_n\in \mathcal{W}_n} && J_n^p(w)\\
        & \mbox{s.t. } && h_n^p(w_n) \leq 0\\
        & && G^p(w_n,w_{-n})\leq 0
    \end{alignat}
    \label{eq:penalized_GNG}
\end{subequations}
\begin{proposition}
\label{prop:equivalence}
    Assume $K_n\geq \frac{Z}{3},\,\forall n$. Then, there exists a set of penalties $(M_j)_{j\in\mathcal{N}}$ such that a solution of the penalized game \eqref{eq:penalized_GNG} is a solution of $\Gamma$.
\end{proposition}
\begin{proof}
    We call the two Lagrangian functions associated to the problems $\mathcal{L}_n^P$ and $\mathcal{L}_n^U$, and we consider the solution for $U$: $(w_n^*,\Theta_n^*)\in\mbox{SOL}\{\mbox{KKT of }U\}$, where $\Theta$ is the set of Lagrangian multipliers. This solution must respect all the constraints and be a solution of $\nabla\mathcal{L}^U(w^*_n,\Theta^*_n)=0$. We now call the part of Lagrangian multipliers that are still present in $P$ $\Theta[\Sigma_n]$, and we consider the set of equations coming from $\nabla\mathcal{L}^P(w_n,\Theta_n[\Sigma_n]; M_n)=0$, where $M_n$ is the set of penalties for agent $n$. If we evaluate these equations in $(w_n^*, \Theta_n^*[\Sigma_n])$, we can solve them w.r.t. the penalties $M_n$, obtaining $\Bar{M}_n$. Given the assumptions, solving the system $\nabla\mathcal{L}^P(w_n,\Theta_n[\Sigma];\Bar{M}_n)=0$ in $(w_n,\Theta_n[\Sigma])$ provides $\tilde{w}_n$ which is a solution also for $\nabla\mathcal{L}^U(w_n,\Theta_n)=0$.
\end{proof}

\subsection{Existence of a Generalized Nash Equilibrium for $\Gamma$}
\label{sec:existence_NE}
To obtain existence result for the relaxed problem, we have to verify assumptions on the variable domains, the objective functions and the constraint sets.
\begin{assumptions}
\label{ass:existence}
We assume the following: 
    \begin{itemize}
        \item[(a)] $\mathcal{W}_n$ are polyhedral and compact for all $n$;
        \item[(b)] $J_n^p(\cdot,w_{-n})$ is continuously differentiable in $w_n$, $h_n^p(\cdot)$ and $G^p(\cdot,w_{-n})$ are twice continuously differentiable in $w_n$ for all $n$;
        \item[(c)] There exists $w^\star \in \mathcal{W}:=\prod_{n\in\mathcal{N}}\mathcal{W}_n$ such that $h_n^p(w_n^\star)<0,\,\forall n$ and $G^p(w^\star)<0$;
        \item[(d)] The Hessians matrices $\nabla^2 h_n^p(w_n)$ for the $l_n$ local constraints are copositive for all $n$; 
        \item[(e)] The Hessians matrices $\nabla^2 G^p(w)$ for the $r$ global constraints are copositive. 
    \end{itemize}
\end{assumptions}
\begin{theorem}[Existence of a solution for $\Gamma^p$]
\label{th:existence}
    Define $\mathcal{W}^\star\subseteq \mathcal{W}$ such that for any $w\in\mathcal{W}^\star$, the Hessian matrix of the objective function is positive semi-definite. Then, there exists a generalized Nash equilibrium on $\mathcal{W}^\star$ solution of $\Gamma$.
\end{theorem}

\begin{proof}
Assumption~\ref{ass:existence} (a) holds, since $\mathcal{W}_n$ is defined as a polyhedral and compact set, i.e., we can express any condition of the form $x\in[x_a,x_b]$ as a set of inequalities $x\leq x_b,x\geq x_a$, which are defined as intersections between closed half spaces. Further, the objective function is continuous and differentiable everywhere in the domain $\mathcal{W}^\star$, and the constraints functions are  bilinear or trilinear. Thus, they are twice continuous and differentiable over the whole domain $\mathcal{W}^\star$. Assumption~\ref{ass:existence} (b) is then satisfied. Assumption~\ref{ass:existence} (c) requires the existence of an interior point. Considering $w^\star$ that strictly satisfies Eq.~\eqref{eq:demand_c}, \eqref{eq:accepted_c}, \eqref{eq:export_c} and \eqref{eq:core_c}, we immediately obtain that $\lambda_n^z = 0, \mu_{n,k} = 0, \sigma_n = 0, \delta_n = 0,\,\forall k\in\{1,\hdots,K_n\},\,\forall n\in\mathcal{N},\,\forall z\in\mathcal{Z}$. Eq.~\eqref{eq:KKT_market} needs instead to be relaxed to a strict inequality. Finally, to check for Assumptions~\ref{ass:existence} (d)-(e), we need to compute the Hessians matrices associated to the constraints functions. Regarding the local set of constraints $h_n^p(w_n)$, we obtain $\nabla^2 h_n^p(w_n) = 0,\, \forall n\in\mathcal{N},\,\forall w_n\in\mathcal{W}_n$, except for Eq.~\eqref{eq:KKT_market}, for which we obtain the following cases:
\begin{equation*}
x_n^{\top}\nabla_{x_n}^2\left(\nabla_{x_{n,k}^z} L\right)x_n = 
    \begin{cases}
        \pi_{n,k}-\lambda_n^z-\delta_n \geq 0, \quad \mbox{if } z = z(n)\\
        \pi_{n,k}-\lambda_n^z+\sigma_n \geq 0, \quad \mbox{if } z \neq z(n)
    \end{cases}
\end{equation*}
Given that an interior point forces $\lambda_n^z=0,\delta_n=0$ and $\sigma_n=0$, and that the prices $\pi_{n,k}$ are positive by definition, both inequalities are always satisfied.
Considering the only remaining global constraint, Eq~\eqref{eq:export_c}, we obtain a constant positive Hessian 
, which directly satisfies the copositivity assumption. We then apply \cite[Theorem 13]{pang_nonconvex_2011} to the $(N+1)-$players game to obtain the existence of a Nash equilibrium $(w^\star,\rho^\star)$, and from \cite[Section 2.1]{pang_nonconvex_2011} we can conclude that $w^\star$ is a generalized Nash equilibrium for $\Gamma^p$. Finally, using Proposition~\ref{prop:equivalence}, a solution of $\Gamma^p$ is also a solution for $\Gamma$, which concludes the proof. 
\end{proof}
We conclude this section checking the potential structure of the game:
\begin{proposition}
\label{prop:GPG}
    The generalized Nash game described in Eq.~\eqref{eq:GNG} is a generalized potential game (GPG).
\end{proposition}
\begin{proof}
    Using \cite[Def. 2.1]{facchinei_decomposition_2011},  two requirements need to be satisfied for a GNG to be a GPG. The first requirement amounts to have a feasible set that is nonempty, which in our setting implies having enough energy offer to cover the total demand.
    The second requirement is for a potential $P(w)$ function to exist, which in our case is satisfied by using $P(w):=\sum_{n=1}^N J_n(w_n,w_{-n})$.
\end{proof}
From Proposition~\ref{prop:GPG}, we can rewrite Eqs.~\eqref{eq:GNG} as
\begin{subequations}
    \label{eq:GPG}
    \begin{align}
    \max_{w\in W} & \quad \sum_{m=1}^N J_n(w_m,w_{-m}),\\
    \text{s.t.} & \quad h_n(w_n) \leq 0, \forall n,\\
    & \quad G(w_n,w_{-n})\leq 0,
\end{align}
\end{subequations}
which can be solved using a standard nonlinear programming solver. The potential structure is interesting algorithmically, to guarantee the convergence of best-response dynamics.

\vspace{-0.3cm}

\section{Algorithmic Solutions}
\label{sec:algorithms}
We propose two exact approaches to compute GNEs of $\Gamma$, that we compare against MARL simulation which capture in a realistic way the interactions between the producers and the market dynamics. One of the exact approach makes use of the GPG structure of $\Gamma$; thus leading to compute a solution of \eqref{eq:GPG} in an integrated way, relying on a nonlinear mathematical optimization solver (MINOS) which is based on the projected Lagrangian method \cite{noauthor_minos_nodate}. We describe next the other exact approach.

\subsection{Gauss-Seidel Best-Response}
\label{sec:gsbr}
We follow the implementation of the Gauss-Seidel best-response algorithm from \cite{facchinei_decomposition_2011}. We solve the problem in Eq.~\eqref{eq:GNG} for each agent sequentially (using MINOS to compute best-response for each player); stopping the whole process once the stopping criterion is met, i.e., the distance between the players' strategy profiles at two consecutive iterations do not change more than a small constant $\varepsilon$. This is summarized in Alg.~\ref{alg:gauss_seidel}.
\begin{algorithm}[tb]
\caption{Gauss-Seidel Best Response.}
\label{alg:gauss_seidel}
\begin{algorithmic}[1] 
\State Initialize problem with feasible point $w^0=[w_1^0,\dots,w_N^0]$, $k=0$, $\varepsilon>0$
\While{$||w^k-w^{k-1}||>\varepsilon$}
\For{$n\in \{1,\dots,N\}$}
\State Solve regularized Eq.~\eqref{eq:GNG} for agent $n$ with $[w_1^{k+1},\dots,\allowbreak w_n^k,\dots,\allowbreak w_N^{k+1}]$
\EndFor
\State Update regularization parameter
\State $k\leftarrow k+1$
\EndWhile
\end{algorithmic}
\end{algorithm}

From \cite[Theorem 5.2]{facchinei_decomposition_2011} and the GPG structure of the game, the Gauss-Seidel best-response algorithm is proved to converge in finite time to a GNE of $\Gamma$.


\subsection{Multi-Agent Deep Reinforcement Learning}
\label{sec:marl}
The core idea of reinforcement learning is to have an agent learning directly from exposure to an unknown environment. It is formally modeled as a Markov Decision Process (MDP) $\{S,A,r\}$ where $S$ is the state space, $A$ the action space and $r:S\times A\to\mathbb{R}$ the reward function. We define the discounted return over a trajectory of length $L$ as $R(s,a) := \sum_{t=0}^{t=L-1}\gamma^tr(s_t,a_t)$, where $0 < \gamma < 1$ is the discount factor. An agent interacts with the environment sampling action from its policy $a\sim\mu(s)$, and it aims to maximize the reward, so the goal is to find the policy that maximizes the discounted return over an infinite horizon. In particular, defining the $Q$-function as:
\begin{equation*}
    Q_{\mu}(s,a) = \mathbb{E}[R(s,a)|\mu,s_0=s,a_0=a],
\end{equation*}
we want to solve the following optimization problem to obtain the optimal policy $\pi^\star$
\begin{equation*}
    Q^\star(s,a) = Q_{\mu^\star}(s,a) = \max_{\mu} Q_{\mu}(s,a).
\end{equation*}
Generally, we do not have access to analytical forms for neither $Q$ nor $\mu$, making the optimization extremely complex to solve. RL algorithms, and in particular \textit{deep} RL (DRL) algorithms, try to solve this problem learning a policy that approximates the optimal one. DRL employs neural networks to model both the policy $\mu$ and the $Q$-function, so that the resulting algorithm is model free and can be used for continuous or discrete states and actions. We will not explore in depth here the differences between DRL algorithms, but we will provide a quick overview of the one we use.

We implemented the multi-agent deep reinforcement learning algorithm described in \cite{lowe_multi-agent_2017}: this is called an \textit{actor-critic} algorithm, where an actor learns the policy and a critic the $Q$-function. The actor and critic are neural networks parameterized by $\theta$ and $\omega$ respectively, which are trained over specific losses. In particular the loss function for the critic is based on the Bellman's equation:
\begin{equation*}
    \mathcal{L}_{Q} = \mathbb{E}_{(s,a,r,s')\sim\mathcal{D}}\left[ (Q_{\omega}(s,a)-(r+\gamma Q_{\omega}(s',\mu_{\theta}(s'))))^2 \right],
\end{equation*}
where $\mathcal{D}$ is a buffer where we save previously encountered states, the actions that the agent used and the reward it obtained.
The actor instead minimizes the following loss:
\begin{equation*}
    \mathcal{L}_{\mu} = -\mathbb{E}_{s\sim\mathcal{D}}\left[Q_{\omega}(s,\mu_{\theta}(s))\right]
\end{equation*}
The intuition behind this algorithm is that the critic learns to map future rewards to specific states and actions (which is the $Q$-function), while the actor tries to maximize the $Q$-function itself (represented by the critic output). In the multi-agent settings we apply a modified version of the single agent implementation presented here, where each agent has its own actor and critic, but the critics receive as input the actions \textit{of all the agents}. This was introduced in \cite{lowe_multi-agent_2017} as a way of stabilizing the learning and improving the results. In particular, given the zonal nature of our environment, we share the actions between the agents of the same zone, but not between all agents.

In our environment, the action $a$ of each producer corresponds to its bids (prices and capacities), while the state $s$ contains the previous day clearing prices, information about the current day and a random zonal perturbation. The reward is a function of the profit and clearing price. In order to run the simulation, each producer outputs their bids, which are used to clear the market solving problem~\eqref{eq:market_clearing}. The results are then used to compute the reward for each producer, which in turn is used to update their strategy.
In Alg.~\ref{alg:zMADDPG} we report in pseudo-code the structure for the training of the RL agents.
\begin{algorithm}[h]
\caption{Basic structure of the zonal MARL algorithm.}
\label{alg:zMADDPG}
\begin{algorithmic}[1] 
\State Initialize problem state $s_0$
\For{$t\in L$}
\State Get action $a_n^t\leftarrow\mu_n(s_t), \forall n \in \mathcal{N}$ 
\State Clear market and obtain $(r^t_n,s_{t+1})\;\forall n\in\mathcal{N}$
\State Build zonal action: $a(z)^t=[a_n^t\;\forall n\in\mathcal{N}(z)]$
\State Store $(s_t,a(z)^t,r^t_n,s^{t+1})$ for each agent
\For{$z\in\mathcal{Z}$}
\For{$n\in\mathcal{N}(z)$}
\State Sample $(s,a(z),r_n,s')$ from $\mathcal{D}$
\State Train critic $Q_n$ with $\mathcal{L}_Q(s,a(z),r_n,s')$
\State Train actor $\mu_n$ with $\mathcal{L}_{\mu}(s,Q_n)$
\EndFor
\EndFor
\EndFor
\end{algorithmic}
\end{algorithm}

\section{Simulation Results}
\label{sec:simul}
We now discuss the numerical results of the three algorithmic solutions we propose. We test these algorithms on real data from the German and Austrian ancillary services markets, for the period between January 2024 and August 2024 \cite{noauthor_accueil_nodate}. The maximum length for the bid sequence is 5, and each submitted capacity must be $\geq 5$ MW. As shown in Tab.~\ref{tab:demands}, the demand in Germany is much higher and volatile than that in Austria. Furthermore, the export limit is fixed at $80$ MW for both zones and the core portion is $100$ MW for Austria and $0$ MW for Germany. The symmetric and low export limit suggests that producers in Germany could have some impact in Austria, while the converse situation is less likely to arise.
\begin{table}[h]
\caption{Demand data summary.}
    \centering
    \begin{tabular}{lcccc}
    \toprule
    Zone & Min & $25$th perc. & $75$th perc. & Max \\
    \midrule
    Germany & $1745$ MW & $1898$ MW & $1988$ MW & $2103$ MW\\
     Austria & $200$ MW & $200$ MW  & $200$ MW & $225$ MW \\
     \bottomrule
    \end{tabular}
    \label{tab:demands}
\end{table}
In our settings we consider 8 producers, with the parameters described in Tab.~\ref{tab:producers}. The parameters are chosen in order to have a competitive environment, making it more difficult to have one or two producers dominating the whole market. We also observed from the real bids submitted, that the total offered capacity should be close to double the demand; our parameters allow for this to happen.
\begin{table}[h]
\caption{Producers' parameters.}
\centering
\begin{tabular}{cccc}
\toprule
Producer & Zone & $\bar{\Delta}$ (MW) & $\underline{\pi}$ \\ \midrule
0 & 0 & 700 & 7 \\
1 & 0 & 700 & 7 \\
2 & 1 & 150 & 3 \\
3 & 1 & 150 & 3 \\
4 & 0 & 650 & 6 \\
5 & 0 & 600 & 5 \\
6 & 0 & 850 & 8 \\
7 & 1 & 350 & 4 \\ \bottomrule
\end{tabular}
\label{tab:producers}
\end{table}
Finally, we test multiple values for the export constraints, in order to study the effect that the two zones have on each other.
All simulations were done on a MacBook Pro with an M2 Pro chip and 16G of memory (code available at \url{https://github.com/FrancescoMorri/MARL_ancillary_markets}).

\subsection{Convergence of the Multi-Agent RL Simulation}
\begin{figure*}[t]
    \centering
    \includegraphics[width=0.8\linewidth]{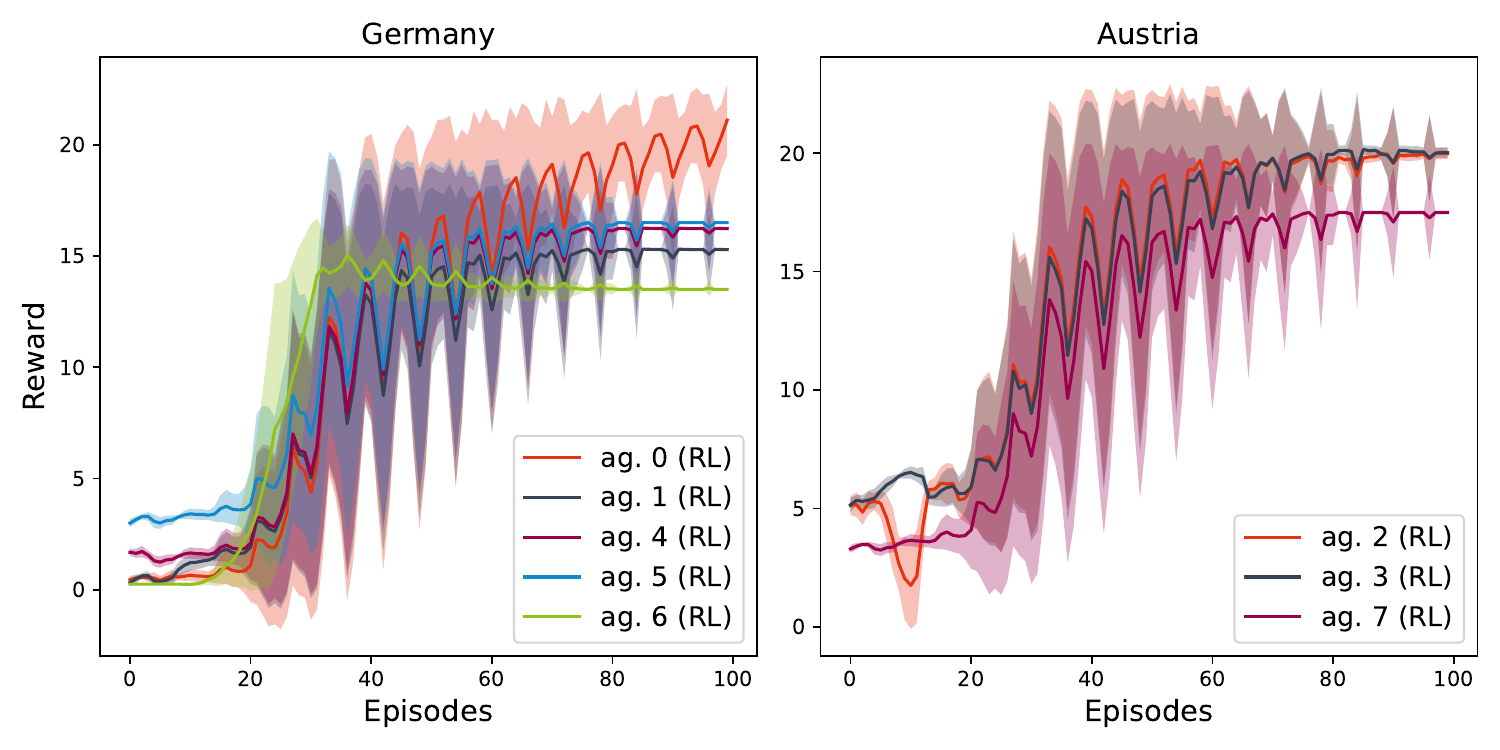}
    \caption{Episodic reward for each RL based agent throughout the training. Each episode represents passing through the full market data. We plot the rewards divided by zone, each point represent the average reward during an episode and the shaded region the standard deviation.}
    \label{fig:8MARL_reward}
\end{figure*}
We start by analyzing the convergence and stability of the multi-agent simulation. Our settings are particularly complicated, since the market environment is dependent on all agents decisions. We tested two settings for MARL algorithm: one with 4 learning based producers, and one with 8. In the first case, the 4 other producers always bid their marginal price. The second case creates a very dynamic environment, since at every step the price of the market is decided based on the actions of 8 independent RL agents. To assess the stability and convergence, we consider the reward of each agent throughout the training. More precisely, we compute the average reward during an episode, which refers to the amount of steps needed to go over the full dataset \cite{noauthor_accueil_nodate}. With this data we obtain Fig.~\ref{fig:8MARL_reward}, where we plot the episodic average reward for each agent divided by zone, considering the case with only RL based producers. Notice that the reward is not the profit, as we mentioned in Section~\ref{sec:marl}. From Fig.~\ref{fig:8MARL_reward} we can see that each agent increases its reward during the initial episodes and eventually plateaus to a stable value. These results are particularly interesting, since they show how the algorithm is well suited for a medium scale simulation with competitive learning agents, providing results that we can compare with the classical solutions.

\subsection{MARL Solution and Generalized Nash Equilibria}
We first discuss the theoretical difference between the integrated solution and the Gauss-Seidel best-response solution.
The two approaches may yield different equilibria: the integrated potential formulation leads to a v-GNE, whereas the best-response produces a GNE that may not be a v-GNE\footnote{Indeed, the set of v-GNEs solutions of $\Gamma$ is included in the set of GNE solutions of $\Gamma$.}. In case a v-GNE is reached, the market allocates the demand in Eq.~\eqref{eq:demand_c} through a single price system, therefore leading to the determination of one price $\lambda$ for the constraint, obtaining a pay-as-clear market. On the contrary, in case there is no market to determine the price system associated with the coupling constraints, two producers $n$, $m$ might attribute different evaluations to the same constraint. This can lead to different valuations of the dual variables, in particular $\lambda_n \neq \lambda_m$, creating a virtual pay-as-bid market.

The solution obtained via MARL simulation is stationary with respect to the reward functions used by the agents. There is no guarantee that this is an exact equilibrium, but the aim is to simulate the agent interactions under varying information-sharing scenarios to analyze strategic behaviors and their impacts on market dynamics.

We first analyze the average cost to run the market, computed by summing the average profits of each producer over the full data available. In particular, for MARL producers, we consider the cost for the final episode of their training, where we assume they reached stable bidding strategies. The results are reported in Fig.~\ref{fig:avg_cost}. The two exact algorithms produce higher costs than the learning-based counterpart, which can be explained by the form of the constraints, as detailed above, and by the lack of actual competition. In the integrated solution, the problem is solved by an exogenous external coordinator, so by definition there will be no competition; while for the best response, each producer is free to update its decision variables uncoupled from the others. We can see the average cost for MARL simulation with only 4 learning producers is very low, but this is also affected by the presence of static producers that always bid at marginal cost, driving the electricity price down.
\begin{figure}[h]
    \centering
    \includegraphics[width=0.9\linewidth]{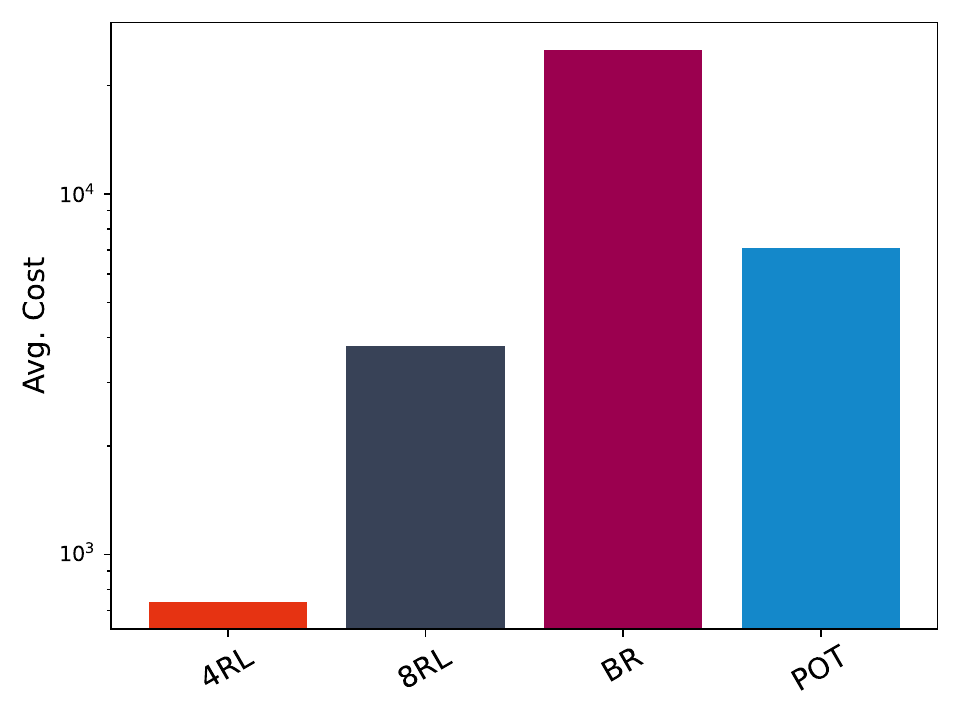}
    \caption{Average cost of running the market, obtained with different algorithms.}
    \label{fig:avg_cost}
\end{figure}
It is interesting to analyze how the profit is shared between producers in the different solutions. To that purpose, we rely on the Gini index $\mathscr G$ \cite{gini_measure_1936}. Given a finite sequence of profits ranked in increasing order $J^{(1)}\leq J^{(2)}\leq\dots\leq J^{(n)}$, where $J^{(l)}$ denotes the $l$-th highest profit, $\mathscr G$ can be obtained as:
\begin{equation*}
    \mathscr G \defeq \frac{1}{n}\left(n+1-2\left(\frac{\sum_{i=1}^n(n+1-i)J^{(i)}}{\sum_{i=1}^n J^{(i)}}\right)\right).
\end{equation*}
The lower $\mathscr G$ is, the more equal is the sharing of profits. We report the results of this analysis in Tab.~\ref{tab:gini}. We can see that, even if the BR and potential solutions produce higher costs, they also obtain a better overall equality index, with the two MARL based simulations getting the highest index. Once again this can be explained by the fact that the multi-agent simulation is unregulated: each learning agent can bid anything it wants, which creates a situation where few producers can drive others' profit almost to zero, using their market power. Looking at the index computed by zone, the differences between classical and learning-based algorithm stay the same.
\begin{table}[h]
\caption{Gini index for all simulations.}
\centering
\begin{tabular}{lccc}
\toprule
Simulation &  Overall $\mathscr G$ &  Germany $\mathscr G$& Austria $\mathscr G$\\ \midrule
Best-Response & $0.43$ & $0.16$ & $0.15$  \\
Potential Solution & $0.51$ & $0.47$ & $0.28$\\
4 RL & $0.64$ & $0.52$ & $0.67$ \\
8 RL & $0.77$ & $0.69$ & $0.33$\\ \bottomrule
\end{tabular}
\label{tab:gini}
\end{table}

We conclude with a comparison of the computational time required by each algorithm to simulate the market for the full dataset. In Tab.~\ref{tab:comp_times} we report the evaluation and training times: the two classical algorithms obviously have no training time, while this is the most time-consuming part for MARL. Further, MARL achieves faster evaluation time than the exact methods. The BR algorithm's evaluation time is longer than the training of 8 RL agents, due to the sequential calls of the optimization solver.
\begin{table}[h]
\caption{Algorithms time for evaluation and training.}
    \centering
    \begin{tabular}{lcc}
    \toprule
    Algorithm & Evaluation Time & Training Time \\
    \midrule
    Best-Response & $\sim 8$h & -\\
    Potential & $\sim 10$m & -\\
    4 RL & $\sim 3$s & $\sim 1.5$h\\
    8 RL & $\sim 4$s & $\sim 2$h\\
    \bottomrule
    \end{tabular}
    \label{tab:comp_times}
\end{table}

\begin{figure*}[t]
    \centering
    \includegraphics[scale=0.5]{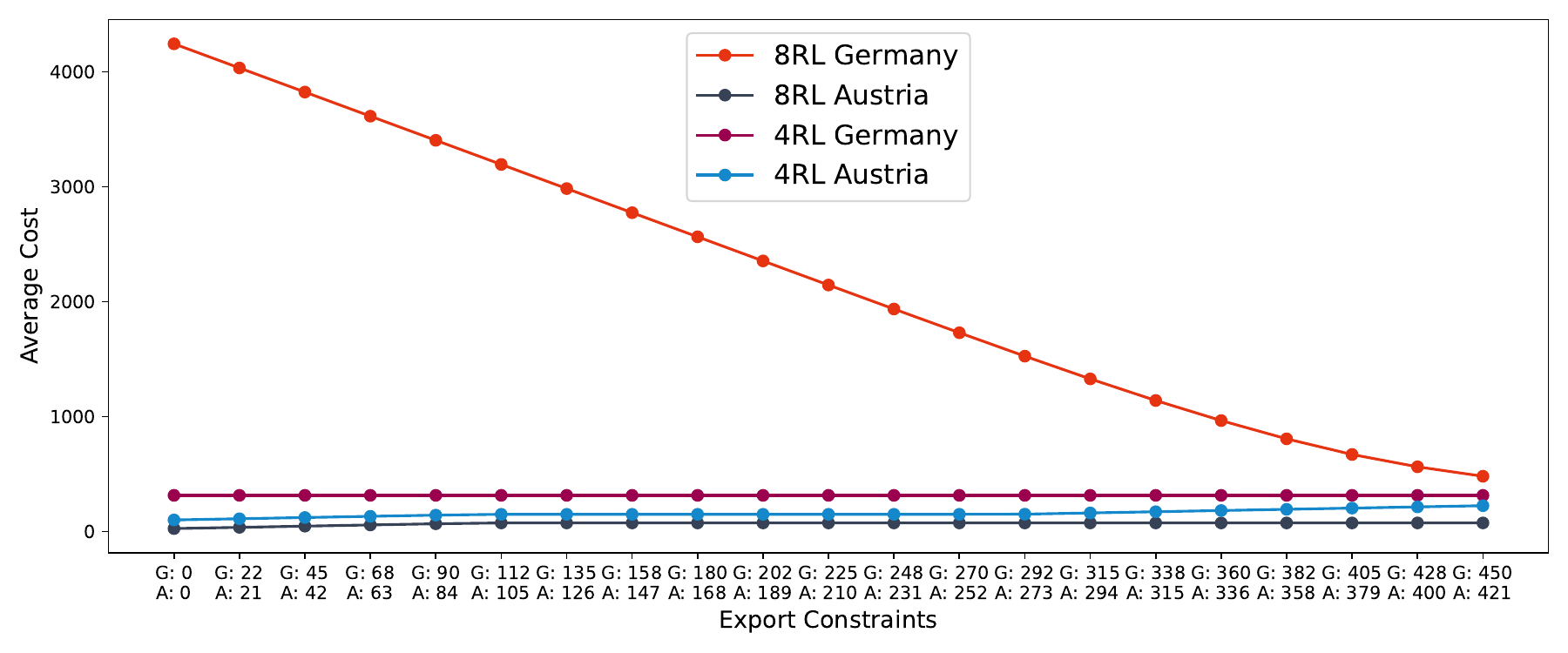}
    \caption{Average cost of running the market with 8 and 4 RL agents, evaluated over different values for the export constraints. Each element on the x-axis corresponds to a couple of export constraints: \texttt{G} stands for Germany and \texttt{A} for Austria, the constraints are in MW.}
    \label{fig:avg_cost_export}
\end{figure*}

\subsection{Effects of Coupling Between Zones}
We now analyze how the two zones influence each other in our model. To do so we simulate the market with a range of different values for the export limit $E_z$ in Eq.~\eqref{eq:export_c}. We write the export limit for Germany (Austria) as a fraction of the maximal demand of Austria (Germany): $E_G = c_G\cdot\max{D_A}$, and we explore the ranges $c_G\in[0,2]$ and $c_A\in[0,0.2]$, thus we go from uncoupled zones to a scenario where the larger zone (Germany) has an export limit which is double the total demand of the smaller zone (Austria).

To quickly evaluate all of these values we use the MARL simulations, in both cases (4 and 8 RL agents). We report in Fig.~\ref{fig:avg_cost_export} the average cost to run the market, computed in the same way as for Fig.~\ref{fig:avg_cost}, for the range of constraints values. In particular, we trained the agents with $c_G=0.4,c_A=0.04$, and used them on all other values. 
We can see that in the case of $4$ RL agents, the average cost is not influenced by the export constraint: this is likely due to the fact that the presence of static agents, always bidding at marginal price, already drives the cost down, and it is not possible to decrease it further. The case with $8$ RL agents is much more interesting: we observe that the cost for Austria is rather stable, while the one for Germany is basically linearly decreasing. We can interpret this behavior considering the differences between the producers in Austria and Germany: in the smaller zone we have less producers, with lower capacities and marginal prices, while in the larger Germany the producers have higher marginal prices. Once the Austrian producers can enter the German market with their leftover capacity, they can bid the same way they would bid for their local market, which in turns, forces the German producers to bid lower, in order to be more competitive. Furthermore, since it is the market that decides how to divide each bid between zone, a producer can not create a priori a specific bid for a specific zone, each bid needs to be \textit{general} enough to be competitive on every zone available.
We conclude by mentioning \cite{di_cosmo_welfare_2020}, where the authors simulated a more economically sophisticated model of the European electricity market, to test what the addition of another smaller zone would do to the costs in the French zone. Their results show the same trend as ours, with the smaller zone obtaining similar cost, while the larger one decreases its own. The main advantage, using MARL based simulation, is that we can quickly test a range of different settings, not just a single scenario, providing a less detailed but broader view of the effect of different market designs. Further, using pretrained agents, we are testing how a bidding strategy developed for a certain market design translates to a different one, when there is no time to adapt. This is relevant for scenarios where the coupling may change due to grid failures, with no warning.


\section{Conclusions}
\label{sec:conclusions}
In this work, we rely on an adversarial Stackelberg game model of the multi-zonal ancillary markets, to study the impact of zonal coupling, that we recast as a nonconvex generalized Nash game with side constraints involving an associated price clearance to be satisfied by the equilibria. Further, relying on the potential structure of the generalized Nash game, we show the existence of a Nash equilibrium. Equilibrium can be reached making use of the game structure and relying on reformulations. Specifically, taking advantage of the potential structure of the game, we compare two exact approaches: Gauss-Seidel best-response and an integrated optimization reformulation. Another paradigm relies on the implementation of a multi-agent reinforcement learning (MARL) framework to simulate decentralized agent interactions, allowing us to study strategic behaviors and their impact on market dynamics. We validate the three algorithmic approaches using real data from the ancillary market between Germany and Austria.
Our results show that, computationally, the MARL simulation (counting training and evaluation time) is more efficient than the best-response method, while remaining a decentralized approach.
The two exact algorithms produce solutions with significantly higher total market costs compared to MARL; although the latter results in a more unfair allocation of profits among players. Numerically, we analyze different levels of coupling between zones, highlighting that tighter coupling may reduce costs in the larger zone; a result which is consistent with findings reported in the literature, e.g., in \cite{di_cosmo_welfare_2020}.


\appendices



\bibliographystyle{IEEEtran}
\bibliography{bibliography}

\newpage

\vfill

\end{document}